\documentclass{amsart}

\usepackage{graphicx}
\usepackage{latexsym,color}
\usepackage{graphicx}
\usepackage{amssymb}
 \usepackage{amsfonts}
 \usepackage{amsmath}
\usepackage{amsthm}
\usepackage[final]{changes}
\newtheorem{thm}{Theorem}[section]
 \newtheorem{dfn}[thm]{Definition}
 \newtheorem{lem}[thm]{Lemma}
 \newtheorem{rem}[thm]{Remark}
\newtheorem{cor}[thm]{Corollary}
\newtheorem{prop}[thm]{Proposition}

\numberwithin{equation}{section}

\begin{document}

\title{The Value of Insider Information for Super--Replication
with Quadratic Transaction Costs }
\thanks{$^*$Corresponding author}

   \author{Yan Dolinsky$^{*}$ \address{
 Department of Statistics, Hebrew University of Jerusalem, Israel. \\
 e.mail: yan.dolinsky@mail.huji.ac.il}}${}$\\
\author{Jonathan Zouari \address{
 Department of Statistics, Hebrew University of Jerusalem, Israel.\\
 e.mail: jonathan.zouari@mail.huji.ac.il}
  ${}$\\
  ${}$\\
Hebrew University }

\date{\today}

\begin{abstract}
We study super--replication of European contingent claims in an illiquid
  market with insider information. Illiquidity is captured by quadratic transaction costs
  and insider information is modeled by an investor who can peek into the future. Our main result
  describes
  the scaling limit of the super--replication prices when the
  number of trading periods increases to infinity. Moreover, the scaling limit gives us the asymptotic value of being an insider.
\end{abstract}

 \maketitle \markboth{Y.Dolinsky and J.Zouari}{The Value of Insider Information}
\renewcommand{\theequation}{\arabic{section}.\arabic{equation}}
\pagenumbering{arabic}

\section{Introduction}
\label{intro}
This paper deals with super--replication of European options with insider information
and quadratic transaction costs.
 We model insider information
by assuming that the investor is allowed to peek $\epsilon$
into the future.
This can be viewed as a progressive enlargement of filtrations.
Although the impact of insider information on utility maximization was largely studied (see, for instance
\cite{ABS,ADI,APS,I,PK}),
to the best of our knowledge
the impact of insider information on super--replication was not studied at all.

Our modeling of illiquidity via quadratic transaction costs
follows the approach of \c{C}etin, Jarrow and
Protter (see \cite{CJP}) who model illiquidity via supply curves.
Supply curves correspond to the case where the size of the
trade has an immediate but temporary effect on the price of the asset.
A linear impact results in quadratic transaction costs. For more details, see the survey paper \cite{GRS}.

In setups with quadratic transaction costs, continuous--time
models are known to produce no liquidity effect at all, see \cite{BB} and \cite{CJP}.
The intuition behind this result is that due to density arguments
the investor can restrict himself to continuous and bounded variation trading strategies,
and for this type of strategies the quadratic transaction costs are zero.
Hence, in the continuous time models, a prophet investor can super--replicate a claim with arbitrary small amount. Namely,
the super--hedging price is minus infinity.

In order to avoid the above "unreasonable" setup, we
consider the
discrete--time version of the
quadratic transaction costs model.
Our sole assumption on the price dynamics is that
the absolute value of the
returns is bounded both from below and above. This can be viewed as the
discrete--time Bachelier version of the widely studied uncertain volatility
models, (see, e.g., \cite{DM1,NN,NS,PRT,STZ}).
We study the scaling limit of the
corresponding super--replication prices when the number of trading periods goes to infinity and the time step goes to zero (continuous--time limit). Moreover, we
assume that the future time interval that the investor can predict tends to zero linearly in the time step.

Our main result (Theorem \ref{thm2.2}) is the characterization of the scaling limit
as a stochastic volatility control problem
defined on the Wiener space.
The main result provides a non-trivial conclusion (Corollary \ref{cor2.1})
which gives the asymptotic value of insider information.

We divide the proof of Theorem~\ref{thm2.2}
into two main steps, namely proving the lower bound (Proposition \ref{prop.campi1}) and proving the upper bound
(Proposition \ref{prop.campi2}).

For the lower bound, we identify the correct
dual representation for the super--replication prices,
and study the limit of the dual terms and its dependence on the insider information.

For the second step, namely, the upper bound, our
main idea is to identify, in the presence of insider information,
a "good" subclass of trading strategies.
Good subclass is understood in the sense that from one hand it does not change the
asymptotic behavior of the super--replication prices and on the other hand
gives a comfortable dual representation, which allows to obtain tightness of the corresponding
consistent price systems.
This idea can be seen as the extension of
the technique recently developed in
\cite{BD},
for the setup of insider trading
and model uncertainty.

Our main result
is not only a generalization of Theorem 2.7 in \cite{BDG} for the setup of insider information, but also an improvement of this theorem.
In \cite{BDG},
in order to get tightness of the consistent price systems (needed for the proof of the upper bound), one had to impose
linear growth constraints on transaction costs and only allowed
quadratic costs in an even smaller region around zero. In this paper we overcome this technical assumption by identifying
a "good" subclass of constrained trading strategies. So Theorem \ref{thm2.2} contains a contribution for the "usual" (i.e.\ no insider information) case as well.

A natural extension of the current work is to consider the more realistic setup
where the insider trading strategy has a permanent impact on the stock prices. Intuitively, if we consider an investor who can pick only small time $\epsilon$ into the future,
such price impact will result in a prohibitively costly super--replication prices. One
way to overcome this is to follow \cite{BD} and consider scaling limits where the permanent price impact tends to zero. However, as we showed in
\cite{BD}
for the setup without
insider information, the asymptotic behaviour of such models is the same as models
with quadratic transaction costs. Hence, a more interesting and challenging direction will be to combine the present methodology of asymptotic analysis together
with the Kyle--Back equilibrium modeling of insider trading (see \cite{B,BP,CCD,DA}).
We
leave this type of extension for future research.

The rest of the paper is organized as follows. In the next section we introduce
the setup and formulate the main results. In Section~\ref{sec:3} we prove the lower bound.
In Section~\ref{sec:4} we prove the upper bound.

\section{Preliminaries and Main Results }\label{sec:2}
Let $0 <\underline{\sigma} \leq \overline{\sigma} <\infty$ be two given constants which represent the volatility bounds.
Consider the sample space
$$\Omega:=\left\{\omega=(x_1,x_2,....)\in\mathbb R^{\mathbb N} : \ \underline{\sigma} \leq  |x_i|\leq \overline{\sigma}, \ \forall i\in\mathbb N\right\}$$
with the Borel $\sigma$--algebra and
the canonical process
$$X_k(\omega):=x_k \ \ \mbox{for} \ \ \omega=(x_1,x_2,....).$$

For any $n\in\mathbb N$ consider a financial market with a riskless savings account and a risky stock. The savings account
will be used as a numeraire and thus we normalize its value to $\equiv 1$. The stock price evolution
$S^{n}=\{S^{n}_k\}_{k=0}^n$  is given by
a Bachelier's type of model
\begin{equation*}
S^{n}_k:=s+\frac{1}{\sqrt n}\sum_{i=1}^k X_i, \ \ k=0,1,...,n
\end{equation*}
where $s\in\mathbb R$ is a constant.

We do not impose a particular probabilistic model and our
sole assumption on the stock price dynamics is that the absolute values of the increments
lie in the interval $[\underline{\sigma}/\sqrt n,\overline{\sigma}/\sqrt n]$.

We fix $N\in\mathbb N$ and consider an investor who can peek $N$ trading periods into the future.
Namely, the filtration of the investor is given by
$\mathcal G_k:=\sigma\{X_1,...,X_{k+N}\}$, $k\in\mathbb Z_{+}$.

Next, let $\Lambda>0$ be a
constant and consider a quadratic cost function $\beta\rightarrow \Lambda \beta^2$ where $\beta$ is
the number of traded stocks. Thus, for the $n$--step model a trading strategy is a map
$\gamma:\{0,1,...,n-1\}\times\Omega\rightarrow \mathbb R$ such that
$\gamma_k:=\gamma(k,\omega)$ is $\mathcal G_k$ measurable function
which is specifying the number of shares held at any period
$k=0,1...,n-1$.

The evolution of the mark--to--market value
$\{Y^{\gamma}_{k}\}_{k=0}^n$ resulting from a trading strategy $\gamma$ is given by
\begin{equation*}
Y^{\gamma}_{k}:=\sum_{i=0}^{k-1} \left(\gamma_i(S^{n}_{i+1}-S^{n}_i)-\Lambda (\gamma_i-\gamma_{i-1})^2\right), \ \ k=0,1,...,n
\end{equation*}
where we set $\gamma_{-1}\equiv 0$.
Hence, we start with zero stocks in our portfolio and
trading to the new position $\gamma_i$ to be held after time $i$ incurs the transaction costs
$\Lambda (\gamma_i-\gamma_{i-1})^2$ which is the only friction in our model. The value $Y^{\gamma}_{k}$
represents the portfolio’s mark--to--market value at time $k$.
Note that, focussing on the mark-to-market value rather than the liquidation
value, we disregard in particular the costs of unwinding any non--zero position for
simplicity.

\begin{rem}
In order to make the analysis less technical we assume that $\underline{\sigma}>0$,
this can be viewed as
assuming that the market is sufficiently volatile.
In this case the stock fluctuation (which might be positive or negative) satisfies
$$|S^n_{k+N}-S^n_k|\geq\frac{\underline{\sigma}N}{n}>0, \ \ k=0,1,...,n-N.$$
Hence, the investor can use his insider information
in order to create an arbitrage opportunities.
However, due to quadratic transaction costs the volume of the trade needs to be bounded.
Thus, although the investor has an arbitrage opportunity, the super--replication
problem still makes sense and the super--replication price is larger than
$-\infty$.
\end{rem}

The main result of this paper is a scaling limit theorem
for the
super--replication costs, letting the number
of trading periods $n$ over the time span $[0,1]$ tend to infinity
while re--scaling the time between trades as $1/n$.

Let $D[0,1]$ be the space
of all RCLL (right continuous with left limits) functions
$p=(p_t)_{t\in [0,1]}:[0,1]\rightarrow\mathbb R$.
Let
  $H:D[0,1] \to \mathbb R$ be a
  continuous map with respect to the Skorohod metric
 \begin{equation*}
  d(p,q) := \inf_{\chi} \left\{\sup_{0 \leq t \leq 1}|t-\chi(t)|+\sup_{0
   \leq t \leq 1} |p_t-q_{\chi(t)}|\right\}, \forall p,q \in D[0,1]
 \end{equation*}
 where the infimum is over all strictly increasing continuous time
 changes $\chi:[0,1] \to [0,1]$ with $\chi(0)=0$ and $\chi(1)=1$.

For any $n\in\mathbb N$ consider a European option in the $n$--step model with the payoff
$Z_n:=H\left(\{S^{n}_{[nt]}\}_{t=0}^1\right)$ where
$[\cdot]$ denotes the integer part of $\cdot$. The process
$\{S^{n}_{[nt]}\}_{t=0}^1$ is viewed as a measurable map from $\Omega$ to $D[0,1]$.
The
super--replication price $\pi_n(Z_n)$ is
then defined as
\begin{equation*}
\pi_n(Z_n):=\inf\left\{y \in \mathbb{R}: \ \exists\gamma
  \text{ with }
y+Y^\gamma_n(\omega)\geq Z_n(\omega) \ \forall \omega\in \Omega\right\}.
\end{equation*}
We emphasize that we require the construction of a robust
super--replication strategy which leads to a terminal value
that dominates the payoff in any conceivable
scenario $\omega \in \Omega$.

Before we arrive at the main results we need some preparations.
Let $\Gamma$ be the set of all bounded (uniformly in time and space), nonnegative progressively measurable processes $\nu=\{\nu_t\}_{t=0}^1$
on the Wiener space $(\Omega^W,\mathcal F^W,(\mathcal F^W_t)_{0 \leq t \leq 1},\mathbb P^W)$
with Wiener process $W$ and its generated filtration $(\mathcal F^W_t)_{0 \leq t \leq 1}$.
For any $\nu\in\Gamma$ we define the process
 $\{S^{\nu}_t\}_{t=0}^1$
\begin{equation}\label{2.last}
S^{\nu}_t :=s + \int_0^t \nu_u \,dW_u, \ \  0 \leq t \leq 1.
\end{equation}
In addition,
introduce the (deterministic) function $G:\mathbb R_{+}\rightarrow\mathbb R_{+}$
\begin{equation*}
    G(z):=
   \begin{cases}
\left(\underline{\sigma}-\frac{z}{\underline{\sigma}}\right)^2, & 0\leq z<\underline{\sigma}^2,\\
0, & \underline{\sigma}^2 \leq z \leq \overline{\sigma}^2,\\
\left(\frac{z}{\overline{\sigma}}-\overline{\sigma}\right)^2, &
z>\overline{\sigma}^2.
   \end{cases}
  \end{equation*}

Now, we are ready to formulate our main results.
We start with the formulation of the lower bound which will be proved in Section \ref{sec:3}.
\begin{prop}\label{prop.campi1}
Let $H:D[0,1] \to \mathbb R$ be a
  continuous function with respect to the Skorohod metric. Moreover, assume that there exists
  constants $C,\mu>0$ such that
  \begin{equation}\label{campi1}
 | H(p)|\leq C\left(1+\sup_{0\leq t\leq T}|p_t|^\mu\right), \ \ \forall  p\in D[0,1].
 \end{equation}
  Then,
 \begin{equation}\label{2.-5}
 \lim\inf_{n\rightarrow\infty} \pi_n(Z_n)\geq \sup_{\nu\in\Gamma}\mathbb E_{\mathbb P^W}\left( H(S^{\nu})-\frac{N}{4\Lambda}|S^{\nu}_1-s|^2-\frac{1}{16\Lambda}\int_{0}^1 G(\nu^2_t) dt\right).
 \end{equation}
\end{prop}
The following simple corollary will be used in the proof of the upper bound.
\begin{cor}\label{cor3.1}
Let $H:D[0,1] \to \mathbb R_{+}$ be a continuous function which satisfies (\ref{campi1}). Then,
$$\lim\inf_{n\rightarrow\infty} \pi_n(Z_n)>-\infty.$$
\end{cor}
\begin{proof}
Follows from choosing $\nu\equiv 0$ in the right hand side of (\ref{2.-5}).
\end{proof}

Next, we formulate the upper bound which will be proved in Section \ref{sec:4}.
\begin{prop}\label{prop.campi2}
Let $H:D[0,1] \to \mathbb R_{+}$ be nonnegative and Lipschitz
  continuous with respect to the Skorohod metric. Then,
\begin{equation*}
\lim\sup_{n\rightarrow\infty} \pi_n(Z_n)\leq\sup_{\nu\in\Gamma}\mathbb E_{\mathbb P^W}\left(H(S^{\nu})-\frac{N}{4\Lambda}|S^{\nu}_1-s|^2-\frac{1}{16\Lambda}\int_{0}^1 G(\nu^2_t) dt\right).
\end{equation*}
\end{prop}
Let us notice that if $H$ is Lipschitz
  continuous with respect to the Skorohod metric then it satisfies the growth condition (\ref{campi1}) (with $\mu=1$).
We now combine the statements of the above propositions and state them as the
main theorem of our paper.
\begin{thm}\label{thm2.2}
Let $H:D[0,1] \to \mathbb R_{+}$ be nonnegative and Lipschitz
  continuous with respect to the Skorohod metric. Then,
\begin{equation*}
\lim_{n\rightarrow\infty} \pi_n(Z_n)=\sup_{\nu\in\Gamma}\mathbb E_{\mathbb P^W}\left(H(S^{\nu})-\frac{N}{4\Lambda}|S^{\nu}_1-s|^2-\frac{1}{16\Lambda}\int_{0}^1 G(\nu^2_t) dt\right).
\end{equation*}
\end{thm}
\begin{proof}
Follows from Proposition \ref{prop.campi1} and Proposition \ref{prop.campi2}.
\end{proof}
\begin{rem}\label{rem1}
For the case $N=0$ (no insider information) our results are improvement
of Theorem 2.7 in \cite{BDG} in the sense that
we do not impose linear growth constraints on the transaction costs.
Moreover, from Proposition \ref{prop.campi1} and Remark \ref{rem3} it follows that
Theorem \ref{thm2.2} can be easily extended to payoffs
of the from
\begin{equation}\label{formm}
\hat H(p):=H(p)-\alpha |p_1-p_0|^2, \ \ \forall  p\in D[0,1]
\end{equation}
where $H$ is as above (nonnegative and Lipschitz continuous) and $\alpha>0$ is a positive constant.
This observation is important for the following
corollary.
\end{rem}
\begin{cor}\label{cor2.1}
Theorem \ref{thm2.2} says that asymptotically the super--replication price
of the European claim $H(S)$ for an insider who can peek $N$ trading periods into
the future equals to the super--replication price for an "usual" investor of the
claim $H(S)-\frac{N}{4\Lambda} |S_1-s|^2$. In other words, for super--replication
with quadratic transaction costs, the asymptotic value of peeking $N$ trading periods into the
future equals to holding (for zero costs) a European option with payoff
 $\frac{N}{4\Lambda} |S_1-s|^2$.
\end{cor}
We end this section with the following remark which discusses what happens if $N$ increases to infinity with $n$.
\begin{rem}
Let $H$ be as in Theorem \ref{thm2.2}.
First we argue that
\begin{equation}\label{campi10}
\lim_{N\rightarrow\infty}\sup_{\nu\in\Gamma}\mathbb E_{\mathbb P^W}\left(H(S^{\nu})-\frac{N}{4\Lambda}|S^{\nu}_1-s|^2-\frac{1}{16\Lambda}\int_{0}^1 G(\nu^2_t) dt\right)=H(\bar s)-\frac{\underline{\sigma}^2}{16\Lambda}
\end{equation}
where  $\bar s\in D[0,T]$ is the constant function $\bar s\equiv s$.
Indeed, by taking $\nu=0$ in the above left hand side we get
$$\lim\inf_{N\rightarrow\infty}\sup_{\nu\in\Gamma}\mathbb E_{\mathbb P^W}\left(H(S^{\nu})
-\frac{N}{4\Lambda}|S^{\nu}_1-s|^2-\frac{1}{16\Lambda}\int_{0}^1 G(\nu^2_t) dt\right)
\geq H(\bar s)-\frac{\underline{\sigma}^2}{16\Lambda}.$$
On the other hand,
from the Doob--Kolmogorov inequality
and the
Lipschitz continuity of $H$ it follows that there exists a constant
$L$ such that
$$\mathbb E_{\mathbb P^W}\left(H(S^{\nu})\right)\leq H(\bar s)+L\mathbb E_{\mathbb P^W}\left(\sup_{0\leq t\leq T}|S^{\nu}_t-s|\right)\leq
 H(\bar s)+2L\left(\mathbb E_{\mathbb P^W}\left(|S^{\nu}_1-s|^2\right)\right)^{1/2}
.$$
This, together with
the It\^o Isometry
$\mathbb E_{\mathbb P^W}\left(|S^{\nu}_1-s|^2\right)= \mathbb E_{\mathbb P^W}\left(\int_{0}^1 \nu^2_t dt\right),
$
and the simple inequality
$G(z)\geq \underline{\sigma}^2-\frac{2\overline{\sigma}z}{\underline{\sigma}}$, $\forall z>0$
gives
$$\lim\sup_{N\rightarrow\infty}\sup_{\nu\in\Gamma}\mathbb E_{\mathbb P^W}\left(H(S^{\nu})
-\frac{N}{4\Lambda}|S^{\nu}_1-s|^2-\frac{1}{16\Lambda}\int_{0}^1 G(\nu^2_t) dt\right)
\leq  H(\bar s)-\frac{\underline{\sigma}^2}{16\Lambda},$$
and (\ref{campi10}) follows.

From Theorem \ref{thm2.2} and (\ref{campi10}) we conclude that
if $N$ increases to infinity with $n$, then asymptotically the super--replication price
will be less or equal than $H(\bar s)-\frac{\underline{\sigma}^2}{16\Lambda}$.
This value does not depend on $\overline{\sigma}$, and so can be viewed as unreasonably low. Thus, roughly speaking, the
"right" scaling is to keep $N$ fixed as $n$ goes to infinity.
\end{rem}
\section{Proof of the lower bound}\label{sec:3}
In this section we prove
Proposition \ref{prop.campi1}.
We start with the following lemma.
\begin{lem}\label{lem3.0}
For any $n$,
\begin{equation}\label{3.0-}
\pi_n(Z_n)\geq\sup_{\mathbb P\in\mathcal P}\mathbb E_{\mathbb P}\left(Z_n-\frac{1}{4\Lambda}\sum_{i=0}^{n-1}\left(\mathbb E_{\mathbb P}\left(S^n_n|\mathcal G_i\right)-S^n_i\right)^2\right)
\end{equation}
where $\mathcal P$ is the set of all Borel probability measures on $\Omega$.
\end{lem}
\begin{proof}
Let $y\in\mathbb R$ be an initial capital and $\gamma$ be a trading strategy
such that
$$y+Y^{\gamma}_n(\omega)\geq Z_n(\omega), \ \ \forall\omega\in\Omega.$$
Then, for any $\mathbb P\in\mathcal P$
\begin{eqnarray*}
&\mathbb E_{\mathbb P}\left(Z_n\right)\leq y+\mathbb E_{\mathbb P}\left(\sum_{i=0}^{n-1} \left(\gamma_i(S^{n}_{i+1}-S^{n}_i)-\Lambda (\gamma_i-\gamma_{i-1})^2\right)\right)\\
&=y+\mathbb E_{\mathbb P}\left(\sum_{i=0}^{n-1} \left((\gamma_i-\gamma_{i-1})(S^n_n-S^n_i)-\Lambda (\gamma_i-\gamma_{i-1})^2\right)\right)\\
&=y+\mathbb E_{\mathbb P}\left(\sum_{i=0}^{n-1}  \left((\gamma_i-\gamma_{i-1})(\mathbb E_{\mathbb P}(S^n_n|\mathcal G_i)-S^n_i)-\Lambda (\gamma_i-\gamma_{i-1})^2\right)\right)\\
&\leq y+\mathbb E_{\mathbb P}\left(\frac{1}{4\Lambda}\sum_{i=0}^{n-1}\left(\mathbb E_{\mathbb P}\left(S^n_n|\mathcal G^n_i\right)-S^n_i\right)^2\right)
\end{eqnarray*}
where the last inequality follows from the simple inequality
$\beta\alpha-\Lambda \beta^2\leq \frac{\alpha^2}{4\Lambda}$, \ \ $\alpha,\beta\in\mathbb R$.
Thus,
$$y\geq \mathbb E_{\mathbb P}\left(Z_n-\frac{1}{4\Lambda}\sum_{i=0}^{n-1}\left(\mathbb E_{\mathbb P}\left(S^n_n|\mathcal G_i\right)-S^n_i\right)^2\right).$$
Since $\mathbb P\in\mathcal P$ was arbitrary we complete the proof.
\end{proof}
\begin{rem}
In fact,
there is an equality in (\ref{3.0-}). Namely, for the case where the payoff $Z_n$ is a continuous function
of $S^n_1,...,S^n_n$ (since $H$ is continuous, this holds true in our setup)
we have
 $$\pi_n(Z_n)=\sup_{\mathbb P\in\mathcal P}\mathbb E_{\mathbb P}\left(Z_n-\frac{1}{4\Lambda}\sum_{i=0}^{n-1}\left(\mathbb E_{\mathbb P}\left(S^n_n|\mathcal G_i\right)-S^n_i\right)^2\right).$$
The proof can be done by following the discretization technique from \cite{BDG} (see Theorem 2.2 there)
or the approach from \cite{CKT} which is based on representation results for increasing convex functionals.
We omit the proof, and provide only the lower bound (this is Lemma \ref{lem3.0}) which is essential for proving Proposition \ref{prop.campi1}.
\end{rem}
Next,
let $C[0,1]$ be the space of all continuous functions
$p:[0,1]\rightarrow\mathbb R$ equipped with the supremum norm
$||p||:=\sup_{0\leq t\leq 1}|p_t|$.
\begin{dfn}\label{dfn1}
Let $\Gamma_0\subset \Gamma$ be the set of all nonnegative, continuous processes $\nu=\{\nu_t\}_{t=0}^1$ defined on the Wiener space
$(\Omega^W,\mathcal F^W,(\mathcal F^W_t)_{0 \leq t \leq 1},\mathbb P^W)$ which are given by
$$\nu_t=f(t,W), \ \ t\in [0,1]$$
 where
$f:[0,1]\times C[0,1]\rightarrow\mathbb R_{+}$
satisfies the following conditions.\\
(i) For any $t\in [0,1]$ and $p,q\in C[0,1]$, if $p_{[0,t]}=q_{[0,t]}$ then
$f(t,p)=f(t,q)$. Namely, $f$ is a progressively measurable map.\\
(ii) The map $f$ is bounded, i.e. $\sup_{(t,p)\in [0,1]\times C[0,1]}f(t,p)<\infty$.\\
(iii) The function $f$ is bounded away from zero, i.e. $\inf_{(t,p)\in [0,1]\times C[0,1]}f(t,p)>0$.
\\
(iv) There is $\delta=\delta(f)>0$ such that $f_{[1-\delta,1]\times C[0,1]}\equiv\overline{\sigma}.$\\
(v) There exists a constant $\mathcal L=\mathcal L(f)$ such that
\begin{equation*}
|f(t_1,p)-f(t_2,q)| \leq \mathcal L\left(|t_1-t_2|+ ||p-q||\right) \ \ \forall (t_1,t_2,p,q)\in [0,1]^2\times C^2[0,1].
\end{equation*}
\end{dfn}
We arrive at the following result.
\begin{lem}\label{lem3.1}
Let $\nu\in\Gamma_0$. There exists a sequence of Borel probability measures $\{\mathbb Q_n\}_{n\in\mathbb N}$ on $\Omega$
such that
we have the weak convergence
\begin{equation}\label{3.1+}
\{S^{n}_{[nt]}\}_{t=0}^1\Rightarrow \{S^{\nu}_t\}_{t=0}^1 \ \mbox{on}  \ D[0,1].
\end{equation}
Namely, the distribution of $\{S^{n}_{[nt]}\}_{t=0}^1$ under $\mathbb Q_n$ converge weakly to the distribution of the process
$\{S^{\nu}_t\}_{t=0}^1$ given by (\ref{2.last}). Moreover,
\begin{equation}\label{3.1++}
\lim_{n\rightarrow\infty}\mathbb E_{\mathbb Q_n}
\left(\sum_{k=0}^{n-1}\left(\mathbb{E}_{\mathbb  Q_n}\left(S^{n}_n\;|\;\mathcal{G}_k\right)-S^{n}_k\right)^2 \right)=
\mathbb E_{\mathbb P^W}\left(N|S^{\nu}_1-s|^2+\frac{1}{4}\int_{0}^1 G(\nu^2_t) dt\right).
\end{equation}
\end{lem}
\begin{proof}
For $n\in\mathbb N$ and a sequence $A=(A_0,...,A_n)\in\mathbb R^{n+1}$ we denote by
$\widetilde{A}$ the continuous linear interpolation in $C[0,1]$ of the points $A_0,...,A_n$. Formally,
$$\widetilde{A}_t:=(nt-[nt])A_{[nt+1]}+(1+[nt]-nt)A_{[nt]}, \ \ t\in [0,1].$$

Let $\nu_t=f(t,W)$, $t\in [0,1]$ where $f$ satisfies conditions (i)--(v) in Definition \ref{dfn1}.
Fix $n\in\mathbb N$ and
define on $\Omega$ the processes $\{\sigma^{n}_k\}_{k=0}^n,\{\kappa^{n}_k\}_{k=0}^n,\{B^{n}_k\}_{k=0}^n$ and $\{\xi^{n}_k\}_{k=1}^n$ by the following recursion:
$$\sigma^{n}_0:=\underline\sigma\vee f(0,0)\wedge \overline\sigma, \ \  \kappa^n_0:=0, \ \ B^{n}_0:=0$$
and, for $k=1,...,n$
\begin{eqnarray*}
&\sigma^{n}_k:=\underline{\sigma}\vee f\left(\frac{k-1}{n}, \widetilde{{B}^{n}}\right)\wedge\overline{\sigma},\\
&\kappa^{n}_k:=\frac{1}{2}\left(\frac{f^2\left(\frac{k-1}{n}, \widetilde{{B}^{n}}\right)}{|\sigma^{n}_k|^2}-1\right),\\
&\xi^{n}_k:=\frac{X_k}{\sigma^{n}_k},\\
&B^{n}_k:=B^{n}_{k-1}+\frac{1}{\sqrt n}\frac{(1+\kappa^{n}_k)X_k-\kappa^{n}_{k-1}X_{k-1}}{\sqrt{1+2\kappa^{n}_k}\sigma^{n}_k}
\end{eqnarray*}
where recall, $X_1,X_2,....$ is the canonical process
on $\Omega$.

From properties (ii)--(iii)
in Definition \ref{dfn1} and the assumption $\underline{\sigma}>0$
it follows that $\kappa^n,\frac{1}{1+2\kappa^n}$ are uniformly bounded in $n$.
Hence,
$B^{n}_k-B^{n}_{k-1}=O(n^{-1/2})$ (unless we state otherwise, the $O$ term is uniform in space and time).
Moreover, the progressive measurability
of $f$ (property (i) in Definition \ref{dfn1})
ensures that its valuation in the definition of $\sigma^{n}_k$ and $\kappa^{n}_k$ depends on $B^{n}$ only via its already constructed values
$B^{n}_0,...,B^{n}_{k-1}$.

Next, we observe that
$\sigma^n_k\in [\underline{\sigma},\overline{\sigma}]$ for all $k$. Thus, there exists
a probability measure $\mathbb Q_n$ on the canonical space $\Omega$ (an atomic measure) such that under $\mathbb Q_n$ the random variables
$\xi^{n}_1,...,\xi^{n}_n=\pm 1$ $\mathbb Q_n$ a.s. and
\begin{equation}\label{3.2}
\mathbb Q_n\left(\xi^{n}_k=\pm 1|\xi^{n}_1,...,\xi^{n}_{k-1}\right):=
\frac{1}{2}\left(1\pm \frac{\kappa^{n}_{k-1}\sigma^{n}_{k-1}\xi^{n}_{k-1}}{\sigma^{n}_{k}(1+\kappa^{n}_k)}\right), \ \ k=1,...,n
\end{equation}
where we set $\xi^n_0\equiv 1$.
From property (v) in Definition \ref{dfn1}, the estimate
$B^{n}_k-B^{n}_{k-1}=O(n^{-1/2})$
and the fact that $\sigma^n$ and $1+2\kappa^n$ are bounded away from zero
it follows that for
sufficiently large $n$ the right hand side of (\ref{3.2}) is in $[0,1]$.
Thus, (for sufficiently large $n$) the probability
measure $\mathbb Q_n$ is well defined.

Now consider the stochastic process $\{M^n_k\}_{k=0}^n$ given by
$$M^{n}_k:=S^{n}_k+\frac{\kappa^{n}_k X_k}{\sqrt n}, \ \ k=0,...,n$$
where we set $X_0\equiv\sigma^{(n)}_0$.
Observe that (\ref{3.2}) yields that
$\{B^{n}_k\}_{k=0}^n$ and $\{M^{n}_k\}_{k=0}^n$ are $\mathbb Q_n$ martingales.
From the definition of $\mathbb Q_n$ and the regularity properties of the function $f$ we have
\begin{eqnarray*}
&\mathbb E_{\mathbb Q_n}\left(|B^{n}_k-B^{n}_{k-1}|^2\left|\right.X_1,...,X_{k-1}\right)\\
&=\frac{1}{n}\frac{(1+\kappa^{n}_k)^2|\sigma^{n}_k|^2-|\kappa^{n}_{k-1}\sigma^{n}_{k-1}|^2}{(1+2\kappa^{n}_k)|\sigma^{n}_k|^2}\\
&=\frac{1}{n}+O(n^{-3/2}).
\end{eqnarray*}
This together with
the Martingale Central Limit theorem (see Proposition 1 in \cite{R})
gives
\begin{equation}\label{3.4}
\{B^{n}_{[nt]}\}_{t=0}^1\Rightarrow W \ \mbox{on} \ D[0,1].
\end{equation}
Notice that
$$M^{n}_k-M^{n}_{k-1}=\sqrt{1+2\kappa^n_k}\sigma^n_k(B^{n}_k-B^{n}_{k-1}), \ \ k=1,...,n.$$
Thus, by combining the stability result Theorem 5.4 in \cite{DP} together with (\ref{3.4}), the regularity of $f$  and the simple estimate
$M^{n}-S^{n}=O(n^{-1/2})$ we obtain
\begin{equation}\label{3.5}
\left(M^{n}_{[nt]},S^{n}_{[nt]}\right)_{t=0}^1\Rightarrow \left(S^{\nu}_t,S^{\nu}_t\right)_{t=0}^1 \ \mbox{on} \ D[0,1]\times D[0,1]
\end{equation}
and
(\ref{3.1+}) follows.

Next, we show (\ref{3.1++}).
From property (iv) in Definition \ref{dfn1} it follows that for sufficiently large $n$, $S^n_n=M^n_n$. Hence,
$$M^n_k=\mathbb E_{\mathbb Q_n}\left(S^n_n|X_1,...,X_k\right), \ \ k=0,1,...,n.$$
This together with the equality $|X_k|^2=\frac{|\sigma^{n}_k|^2}{n}$ $\mathbb Q_n$ a.s., yields
\begin{eqnarray}
&\sum_{k=0}^{n-1} \left(\mathbb{E}_{\mathbb  Q_n}(S^{n}_n\;|\;\mathcal{G}_k)-S^{n}_k\right)^2\nonumber\\
&=\sum_{k=0}^{n-1-N}  \left(M^{n}_{k+N}-S^{n}_k\right)^2+O(1/n)\nonumber\\
&=\frac{1}{n}\sum_{k=0}^{n-1-N} \left(\kappa^{n}_{k+N} X_{k+N}+\sum_{j=1}^N X_{k+j}\right)^2+O(1/n)\nonumber\\
&=\frac{1}{n}\sum_{i=0}^{n-1} |\kappa^{n}_i|^2|\sigma^{n}_i|^2+\frac{N}{n}\sum_{i=0}^{n-1} |\sigma^{n}_i|^2+2\sum_{j=1}^N I^{n}_j+O(1/n)\label{3.6}
\end{eqnarray}
where
$$I^{n}_j:=\frac{1}{n}\left(\sum_{k=0}^{n-1-N} (1+\kappa^{n}_{k+N}) X_{k+N} X_{k+j}+\sum_{k=0}^{n-1-N}\sum_{i=j+1}^{N-1} X_{k+i} X_{k+j}\right).$$

From the continuity and boundedness of $f$ we get the weak convergence (on $\mathbb R$)
\begin{equation}\label{3.7-}
\frac{1}{n}\sum_{i=0}^{n-1} |\kappa^{n}_i|^2|\sigma^{n}_i|^2 \Rightarrow  \frac{1}{4}\int_{0}^1 G(\nu^2_t) dt
\end{equation}
and
\begin{equation}\label{3.7}
\frac{N}{n}\sum_{i=0}^{n-1} |\sigma^{n}_i|^2\Rightarrow N\int_{0}^1 \left(\underline\sigma\vee\nu_t\wedge\overline\sigma\right)^2 dt.
\end{equation}

Next, we fix $j$ and estimate $I^{n}_j$.
Writing $I^{n}_j$ as a telescopic sum we have
\begin{eqnarray*}
&I^{n}_j=\frac{1}{n}\sum_{i=j+1}^{N}\sum_{k=0}^{n-1-N} \left((1+\kappa^{n}_{k+i}) X_{k+i} X_{k+j}-\kappa^{n}_{k+i-1} X_{k+i-1} X_{k+j}\right)\\
&+\frac{1}{n}\sum_{k=0}^{n-1} \kappa^{n}_{k}X^2_{k}+O(1/n)\\
&=\sum_{i=j+1}^{N}\sum_{k=0}^{n-1-N}(S^{n}_{k+j}-S^{n}_{k+j-1})\left(M^{n}_{k+i}-M^{n}_{k+i-1}\right)\\
&+\frac{1}{n}\sum_{k=0}^{n-1} \kappa^{n}_{k}|\sigma^{n}_k|^2+O(1/n).
\end{eqnarray*}
From Theorem 4.1 in \cite{DP} and (\ref{3.5}) we obtain
$$\sum_{k=0}^{n-1-N}(S^{n}_{k+j}-S^{n}_{k+j-1})\left(M^{n}_{k+i}-M^{n}_{k+i-1}\right)
\Rightarrow \int_{0}^1 (S^{\nu}_t-S^{\nu}_t)dS^{\nu}_t=0, \ \ i=j+1,...,N.$$
Hence, for any $j=1,...,N$
$$I^n_j\Rightarrow \frac{1}{2}\left(\int_{0}^1 \nu^2_t dt-\int_{0}^1\left(\underline\sigma\vee\nu_t\wedge\overline\sigma\right)^2 dt\right).$$
This together with (\ref{3.6})--(\ref{3.7}) and the uniform boundedness of $\sigma^{n},\kappa^{n}$ gives
(\ref{3.1++}).
\end{proof}
We now have all the pieces in place that we need for the completion of the proof of the lower bound.\\
${}$\\
\textbf{Proof of Proposition \ref{prop.campi1}.}\\
\textbf{First Step:}
In this step we prove that (recall that $\Gamma_0\subset\Gamma$)
\begin{eqnarray*}
&\sup_{\nu\in\Gamma}\mathbb E_{\mathbb P^W}\left(H(S^{\nu})-\frac{N}{4\Lambda}|S^{\nu}_1-s|^2-\frac{1}{16\Lambda}\int_{0}^1 G(\nu^2_t) dt\right)\\
&=\sup_{\nu\in\Gamma_0}\mathbb E_{\mathbb P^W}\left(H(S^{\nu})-\frac{N}{4\Lambda}|S^{\nu}_1-s|^2-\frac{1}{16\Lambda}\int_{0}^1 G(\nu^2_t) dt\right).
\end{eqnarray*}
Let $\nu\in\Gamma$. There exists a constant $c>0$ such that $\nu\leq c$ $dt\otimes\mathbb P^W$ a.s.
Using similar density arguments as in Lemma~7.3 in \cite{DS}
gives
that there exists a sequence $\{\nu^{n}\}_{n=1}^{\infty}\subset\Gamma_0$ such that
$\nu_n\rightarrow \nu$ in probability (with respect to $dt\otimes\mathbb P^W$) and for any
$n$ we have $\nu^n\leq c$ $dt\otimes\mathbb P^W$ a.s.

From the growth condition (\ref{campi1})
we obtain that the random variables
$$H(S^{\nu^n})-\frac{N}{4\Lambda}|S^{\nu^n}_1-s|^2-\frac{1}{16\Lambda}\int_{0}^1 G(|\nu^n_t|^2) dt, \ \ n\in\mathbb N$$ are
uniformly integrable, and so from the continuity of $H$
\begin{eqnarray*}
&\mathbb E_{\mathbb P^W}\left(H(S^{\nu})-\frac{N}{4\Lambda}|S^{\nu}_1-s|^2-\frac{1}{16\Lambda}\int_{0}^1 G(\nu^2_t) dt\right)\\
&=\lim_{n\rightarrow\infty}\mathbb E_{\mathbb P^W}\left(H(S^{\nu^n})-\frac{N}{4\Lambda}|S^{\nu^n}_1-s|^2-\frac{1}{16\Lambda}\int_{0}^1 G(|\nu^n_t|^2) dt\right)
\end{eqnarray*}
and the first step is completed.
\\
\textbf{Second Step:}
In view of the first step, in order to complete the proof of Proposition \ref{prop.campi1}
 it remains to show that for any $\nu\in\Gamma_0$
 \begin{equation*}
 \lim\inf_{n\rightarrow\infty}\pi_n(Z_n)\geq \mathbb E_{\mathbb P^W}\left(H(S^{\nu})-\frac{N}{4\Lambda}|S^{\nu}_1-s|^2-\frac{1}{16\Lambda}\int_{0}^1 G(\nu^2_t) dt\right).
 \end{equation*}
 Choose $\nu\in\Gamma_0$. Consider the probability measures $\mathbb Q_n$, $n\in\mathbb N$ and
the corresponding martingales $M^n$, $n\in\mathbb N$ which are constructed in Lemma \ref{lem3.1}.

Observe that $|M^{n}_k-M^{n}_{k-1}|=O(n^{-1/2})$ for all $k\leq n$. Thus, from Lemma 3.3 in \cite{BDP}
it follows that (recall the constant $\mu$ from Proposition \ref{prop.campi1}).
$$\sup_{n\in\mathbb N}\mathbb E_{\mathbb Q_n}\left(\max_{0\leq k\leq n}|M^n_k|^{2\mu}\right)<\infty.$$

From (\ref{campi1}) and the simple bound $|S^n-M^n|=O(n^{-1/2})$ we obtain that
$H\left(\{S^{n}_{[nt]}\}_{t=0}^1\right)$, $n\in\mathbb N$
are uniformly integrable. This together with (\ref{3.1+}) and the continuity of $H$ gives
\begin{equation}\label{campi2}
\lim_{n\rightarrow\infty}\mathbb E_{\mathbb Q_n}\left(H\left(\{S^{n}_{[nt]}\}_{t=0}^1\right)\right)=
\mathbb E_{\mathbb P^W}\left(H(S^{\nu})\right).
\end{equation}
Finally, from
Lemma \ref{lem3.0},
(\ref{3.1++}) and (\ref{campi2})
it follows that
\begin{eqnarray*}
&\lim\inf_{n\rightarrow\infty} \pi_n(Z_n)\\
&\geq \lim\inf\mathbb E_{\mathbb Q_n}\left(H\left(\{S^{n}_{[nt]}\}_{t=0}^1\right)-\frac{1}{4\Lambda}\sum_{i=0}^{n-1}\left(\mathbb E_{\mathbb P}\left(S^n_n|\mathcal G_i\right)-S^n_i\right)^2\right)\\
&=\mathbb E_{\mathbb P^W}\left(H(S^{\nu})-\frac{N}{4\Lambda}|S^{\nu}_1-s|^2-\frac{1}{16\Lambda}\int_{0}^1 G(\nu^2_t) dt\right)
\end{eqnarray*}
and the proof is completed.
\qed

\section{Proof of the upper bound}\label{sec:4}
In this section we prove Proposition \ref{prop.campi2}.

We start with a discretization of the set $\Omega$.
For any $k\in\mathbb N$
let $\Omega^k$ be the set of all sequences $(x_1,x_2,....)\in\Omega$ such that for any
$i\in\mathbb N$ we have
$$|x_i|=\frac{j}{k}\underline{\sigma}+\left(1-\frac{j}{k}\right)\overline{\sigma}$$
for some $j=j(i)\in\{0,1,...,k\}$.

For any $k,n\in\mathbb N$ consider the $n$--step financial market supported on the set $\Omega^k\subset\Omega$.
Observe that this is a (finite) multinomial model.
For any map $Z:\Omega^k\rightarrow\mathbb R$ we define the corresponding super--replication price
\begin{equation*}
\pi^k_n(Z)=\inf\left\{y \in \mathbb{R}: \ \exists\gamma
  \text{ with }
y+Y^\gamma_n(\omega)\geq Z(\omega) \ \forall \omega\in \Omega^k\right\}
\end{equation*}
where the definition of the trading strategy $\gamma$ and the corresponding wealth process $Y^{\gamma}$ are the same as in Section \ref{sec:2}.
\begin{lem}\label{lem4.1}
For any $n\in\mathbb N$,
$$\lim_{k\rightarrow\infty}\pi^k_n(Z_n)=\pi_n(Z_n).$$
\end{lem}
\begin{proof}
Fix $n$. Clearly $\pi^k_n(Z_n)\leq\pi_n(Z_n)$ for all $k$. Thus, we need to show that
 $\lim\inf_{k\rightarrow\infty}\pi^k_n(Z_n)\geq\pi_n(Z_n)$.
 Since $H:D[0,1]\rightarrow\mathbb R$ is continuous then there exists a continuous function
 $F_n:\mathbb R^{n}\rightarrow\mathbb R_{+}$ such that
 $Z_n=F_n(S^n_1,...,S^n_n)$.

Hence, we can mimic the proof of Lemma 3.3 in \cite{BDG}. The only difference
is that, since in the present setup we have insider information, we need to provide a different argument to the fact that
without loss of generality we can assume that
for any $k$ the corresponding trading strategy
$\gamma^k$ is uniformly bounded in $k$.

To that end choose $k\in\mathbb N$. Let $y$ be an initial capital and $\gamma^k$ be a trading strategy such that
\begin{equation}\label{4.2}
y+Y^{\gamma^k}_n(\omega)\geq Z_n(\omega), \ \ \forall \omega\in\Omega^k.
\end{equation}
Since  $H:D[0,1]\rightarrow\mathbb R$ is continuous then
$||Z_n||_{\infty}:=\sup_{\omega\in\Omega}Z_n(\omega)<\infty$.
Thus, without loss of generality we assume that $y\leq ||Z_n||_{\infty}$. This together with (\ref{4.2}) and the fact that $Z_n\geq 0$ gives
$$||Z_n||_{\infty}+\sum_{i=0}^{n-1}\left(\frac{\overline{\sigma}}{\sqrt n}|\gamma^k_i|-\Lambda  (\gamma^k_i-\gamma^k_{i-1})^2\right)\geq 0,
\ \ \forall \omega\in\Omega^k$$
and so (recall that $\gamma^k_{-1}\equiv 0$), $$\sup_{k\in\mathbb N}\max_{0\leq i\leq n-1}\sup_{\omega\in\Omega^k}|\gamma^k(i,\omega)|<\infty$$
as required.
\end{proof}
From Lemma \ref{lem4.1} it follows that for any $n\in\mathbb N$ there exists $k=k(n)$
such that
$\pi^{k(n)}_n(Z_n)>\pi_n(Z_n)-\frac{1}{n}$.
For simplicity we denote
$\Omega_n:=\Omega^{k(n)}_n$ and
$\hat\pi_n(\cdot):=\pi^{k(n)}_n$, i.e.
\begin{equation}\label{4.new}
\hat\pi_n(Z_n)>\pi_n(Z_n)-\frac{1}{n}.
\end{equation}
Thus, in order to prove Proposition \ref{prop.campi2} we need
to establish the inequality
\begin{equation}\label{4.3}
\lim\sup_{n\rightarrow\infty} \hat\pi_n(Z_n)\leq\sup_{\nu\in\Gamma}\mathbb E_{\mathbb P^W}\left(H(S^{\nu})-\frac{N}{4\Lambda}|S^{\nu}_1-s|^2-\frac{1}{16\Lambda}\int_{0}^1 G(\nu^2_t) dt\right).
\end{equation}
The fact that the super--replication price $\hat\pi_n(\cdot)$ is defined in the multinomial model setup will be used in our duality analysis
in Proposition \ref{prop4.1}.

Our preliminary step in the asymptotic analysis of the
super--replication prices
$\hat\pi_n(Z_n)$, $n\in\mathbb N$
is the space--time discretizations
of the price processes.
Specifically, for any
$\epsilon\in (0,1)$ and $n\in\mathbb N$ define (on $\Omega_n$)
a sequence of stopping times by the following recursive relations: $\tau^{n,\epsilon}_0:=0$ and for
$k\in\mathbb N$
$$\tau^{n,\epsilon}_k:=1\wedge\inf\left\{t \geq \tau^{n,\epsilon}_{k-1}\;:\; |S^n_{[nt]}-S_{[n\tau^{n,\epsilon}_{k-1}]}| \geq \epsilon \text{ or } |t-\tau^{n,\epsilon}_{k-1}| \geq \epsilon^2\right\}.$$
Set,
\begin{equation}\label{4.recall}
S^{n,\epsilon}_t:=\sum_{k\in\mathbb N} S^n_{[n\tau^{n,\epsilon}_{k-1}]}
\mathbb I_{t\in [\tau^{n,\epsilon}_{k-1},\tau^{n,\epsilon}_k)}+S^n_n\mathbb{I}_{t=1}, \ \ 0 \leq t \leq 1.
\end{equation}
As usual $\mathbb I$ denotes the indicator function.

Observe that since $H$ is Lipschitz in the Skorohod metric then
$H(p)$ has a linear growth in $||p-p_0||:=\sup_{0\leq t\leq T}|p_t-p_0|$. Thus, for any $\lambda>0$
there exists $c(\lambda)>0$  such that
\begin{equation}\label{4.4}
H(p) \leq c(\lambda)+\lambda^2||p-p_0||^2, \ \ \forall p\in D[0,1].
\end{equation}
 Without loss of generality we assume that $c(\lambda)>2$.
Set,
\begin{eqnarray*}
&K=K(\epsilon,\lambda):=[c(\lambda)/(\epsilon\lambda)^2]+1\\
&Z^{n,\epsilon,K}:= H(S^{n,\epsilon}) \mathbb I_{\{\tau^{n,\epsilon}_K=1\}}.
\end{eqnarray*}
\begin{lem}\label{lem4.2}
There exists $\lambda_0\in (0,1)$ such that for
any $\epsilon\in (0,1)$ and $\lambda<\lambda_0$,
$$\lim\sup_{n\rightarrow\infty}\left(\hat\pi_n(Z_n)-(1-\lambda)\hat\pi_n(Z^{n,\epsilon,K}/(1-\lambda))\right)\leq O(\epsilon+\lambda).$$
\end{lem}
\begin{proof}
Choose $\epsilon,\lambda\in (0,1)$.
From the fact that $H$ is Lipschitz in the Skorohod metric we obtain that for sufficiently large $n$
\begin{equation}\label{4.5}
Z_n\leq H(S^{n,\epsilon})+O(\epsilon).
\end{equation}
Next, choose $n\in\mathbb N$ sufficiently large such that (\ref{4.5}) holds true.
Set,
$$Q^{n,\epsilon}:=\sup_{0 \leq t \leq 1}|S^{n,\epsilon}_t-s|^2
    +\sum_{k\in\mathbb N}\left(|S^n_{[n\tau^{n,\epsilon}_k]}-S^n_{[n\tau^{n,\epsilon}_{k-1}]}|^2+|\tau^{n,\epsilon}_k-\tau^{n,\epsilon}_{k-1}|\right).$$
From the definition of $\tau^{n,\epsilon}_{k}$, $k\in\mathbb N$ it follows that
\begin{equation}\label{estimate}
K \epsilon^2 \leq \sum_{k=1}^K
   \left(|S^n_{\tau^{n,\epsilon}_k}-S^n_{\tau^{n,\epsilon}_{k-1}}|^2+|\tau^{n,\epsilon}_k-\tau^{n,\epsilon}_{k-1}|\right) \
   \text{on the event} \  \left\{\tau^{n,\epsilon}_K<1\right\}.
   \end{equation}
   This together with (\ref{4.4})--(\ref{4.5}) gives
   $$Z_n\leq Z^{n,\epsilon,K}+\lambda^2 Q^{n,\epsilon}+O(\epsilon).$$
   Convexity of the wealth dynamics implies convexity of the super--replication price. Hence,
   \begin{equation}\label{4.6}
  \hat\pi_n(Z_n) \leq (1-\lambda)\hat\pi_n(Z^{n,\epsilon,K}/(1-\lambda))+\lambda\hat\pi_n(\lambda Q^{n,\epsilon})+O(\epsilon).
  \end{equation}
Following the arguments of Lemma 3.6 in \cite{BD} (replace $\sigma$ with $\overline{\sigma}$)
we get that there exists $\lambda_0>0$ such that for any $\lambda<\lambda_0$ we have
$\hat\pi_n(\lambda Q^{n,\epsilon})\leq \lambda(1+36\overline{\sigma}^2)$.
This together with (\ref{4.6}) completes the proof.
\end{proof}
\begin{rem}
Although Lemma 3.6 in \cite{BD} deals with binomial models, it is straightforward to check that the same arguments will work in our setup,
if we replace $\sigma$ with $\overline{\sigma}$. Clearly, insider information can not increase the super--replication price.
\end{rem}
The
main step in the proof of the upper bound is to understand how to
super-replicate the claims $Z^{n,\epsilon,K}/(1-\lambda)$
in the presence of insider information.
Notice that these claims depend on the values of their underlying at only a fixed number $K$ of sampling times.
Such claims turn out to allow for a particularly convenient duality estimate for their super--replication prices.
The idea is that
rather than looking
on all trading strategies for a cost--effective super-hedge, we
will consider a suitably constraint class.

Fix $n\in\mathbb N$ and $\epsilon>0$.
Let $[S^n]_k:=\frac{1}{n}\sum_{i=1}^k X^2_i$, $k=0,1...,n$ be the quadratic variation of $S^n$
and let $(\mathcal F^{n,\epsilon}_k)_{k\in\mathbb Z_{+}}$ be the filtration given by
$$\mathcal F^{n,\epsilon}_k:=\sigma\left\{\tau^{n,\epsilon}_1,...,\tau^{n,\epsilon}_k, S^n_{[n\tau^{n,\epsilon}_1]},...,
S^n_{[n\tau^{n,\epsilon}_k]},[S^n]_{[n\tau^{n,\epsilon}_1]},...,[S^n]_{[n\tau^{n,\epsilon}_k]}
\right\}.$$
We have the following computational result.
\begin{lem}\label{lem4.3}
Let $m\in\mathbb N$. Denote by
$\mathcal A_m$ the set of all
$(\mathcal F^{n,\epsilon}_k)_{k=0,\dots,m}$--adapted processes $\{\alpha_k\}_{k=0}^{m}$ such that
$|\alpha_k|\leq \log n$ for $k=0,\dots,m$.
Then, for any $\phi,\psi\in\mathcal A_m$ there exists a trading strategy $\gamma$ (in the sense that given in Section \ref{sec:2}) such that
\begin{eqnarray*}
&Y^{\gamma}_n\geq \sum_{k=0}^{m}\phi_k \left(S^n_{[n\tau^{n,\epsilon}_{k+1}]}-S^n_{[n\tau^{n,\epsilon}_k]}\right)\\
&+\sum_{k=0}^{m}\left(\frac{\psi_k}{2}+\frac{N}{4\lambda}\right) \left(S^n_{[n\tau^{n,\epsilon}_{k+1}]}-S^n_{[n\tau^{n,\epsilon}_k]}\right)^2 \nonumber\\
&-\sum_{k=0}^{m}\left(\frac{\psi_k}{2}+\Lambda\psi^2_k\right)\left([S^n]_{[n\tau^{n,\epsilon}_{k+1}]}-[S^n]_{[n\tau^{n,\epsilon}_{k}]}\right)-
O(\log^2 n/n^{1/6})\nonumber
\end{eqnarray*}
The above $O$ term may depend on $\epsilon$ and $m$.
\end{lem}
\begin{proof}
 For $k=0,1,...,m+1$ introduce the random variables (integer valued)
$$a_k:=[n\tau^{n,\epsilon}_k], \ \  b_k:=a_k+[n^{1/3}].$$
For $k=0,1,...,m$ set $c_k:=a_{k+1}-[n^{1/3}].$
Let us define the trading strategy $\gamma$ on each of the time intervals $[a_k,a_{k+1}-1]$,
$k=0,1,...,m$. If $a_k>n-2 n^{1/3}$ we do not trade in the time interval $[a_k,a_{k+1}-1]$.
Otherwise (observe that for sufficiently large $n$ $b_k<c_{k}$), we divide the interval
$$[a_k,a_{k+1}-1]= [a_k,b_k)\cup [b_k,c_k) \cup [c_k,a_{k+1}-1)$$
into 3 disjoint intervals.
On the first interval
 $[a_k,b_k)$ we trade in a constant speed and change the number of shares from $0$ to
 $\phi_k$. On the second interval
$[b_k,c_k)$ we trade such that for any
$i\in [b_k,c_k)$ we buy
$$\psi_k (S^n_{i+N}-S^n_{i+N-1})+\frac{S^n_{i+N}-S^n_i}{2\Lambda}=\frac{1}{\sqrt n}\left(\psi_k X_{i+N}+\frac{\sum_{j=i+1}^{i+N} X_j}{2\Lambda}\right)$$
number of shares. Namely,
\begin{equation}\label{4.7}
\gamma_i-\gamma_{i-1}=\frac{1}{\sqrt n}\left(\psi_k X_{i+N}+\frac{\sum_{j=i+1}^{i+N} X_j}{2\Lambda}\right), \ \ i\in [b_k,c_k).
\end{equation}
Recall, that $N$ is the number of steps that the investor can peek into the future.
Finally, on the last interval
$[c_k,a_{k+1}-1)$
we liquidate our portfolio in a constant speed.

Let us estimate the portfolio wealth of such strategy.
Observe that $S^n$ is uniformly bounded (in $n$) on the time interval
  $[0,a_m]$, and so, from (\ref{4.7}) we conclude that
$\gamma$ is of size $O(\log n)$.
Thus, the transaction costs on the (small) time intervals $[a_k,b_k)$ and $[b_k,c_k)$ are of order $O(\log^2 n/n^{1/3})$.
Next, we notice that on the intervals
$[a_k,b_k)$ and $[c_k,a_{k+1}-1)$ the stock price fluctuation is of order $O(n^{1/3}n^{-1/2})=O(n^{-1/6})$.
Therefore, by the summation by parts formula, (also recall the non trading time interval $[n-2 n^{1/3},n]$)
\begin{eqnarray}
&Y^\gamma_n=\sum_{k=0}^{m}\sum_{i=b_k}^{c_k-1}\gamma_i(S^n_{i+1}-S^n_i)-\Lambda (\gamma_i-\gamma_{i-1})^2-O(\log^2 n/n^{1/6})\nonumber\\
&=\sum_{k=0}^{m}[J^1_k-J^2_k-J^3_k]-O(\log^2 n/n^{1/6})\label{4.7+}
\end{eqnarray}
where
\begin{eqnarray*}
&J^1_k:=\gamma_{c_k}(S^n_{c_k}-S^n_{b_k}),\\
&J^2_k:=\sum_{i=b_k+1}^{c_k}(\gamma_i-\gamma_{i-1})(S^n_i-S^n_{b_k}),\\
&J^3_k:=\sum_{i=b_k}^{c_k-1}\Lambda (\gamma_i-\gamma_{i-1})^2.
\end{eqnarray*}
From (\ref{4.7}), the fact that $S^n$ is uniformly bounded on
  $[0,a_m]$ and the stock price fluctuation on the small intervals is of order $O(n^{-1/6})$
 we get
  \begin{eqnarray}\label{4.8}
&J^1_k=\phi_k (S^n_{a_{k+1}}-S^n_{a_k})+\left(\psi_k+\frac{N}{2\lambda}\right)(S^n_{a_{k+1}}-S^n_{a_k})^2+O(\log n/n^{1/6})\nonumber\\
&=\phi_k (S^n_{a_{k+1}}-S^n_{a_k})+\left(\frac{\psi_k}{2}+\frac{N}{4\Lambda}\right)\left((S^n_{a_{k+1}}-S^n_{a_k})^2+[S^n]_{a_{k+1}}-[S^n]_{a_{k}}\right)\\
&+\frac{1}{n}\left(\psi_k+\frac{N}{2\lambda}\right)\sum_{a_k\leq i<j\leq a_{k+1}}X_i X_j+O(\log n/n^{1/6}).\nonumber
\end{eqnarray}
Next, applying (\ref{4.7}) again we obtain
\begin{eqnarray}\label{4.9}
&J^2_k=\frac{\psi_k}{n}\sum_{a_k\leq i,j\leq a_{k+1}, j-i\geq N}X_i X_j+ \nonumber\\
&\frac{1}{2\Lambda n}\sum_{p=1}^N \sum_{a_k\leq i,j\leq a_{k+1}, j-i\geq p}X_i X_j+O(\log n/n^{1/6})\nonumber\\
&=\frac{1}{n}\left(\psi_k+\frac{N}{2\lambda}\right)\sum_{a_k\leq i<j\leq a_{k+1}}X_i X_j\\
&-\frac{1}{n}\sum_{p=1}^{N-1}\left(\psi_k+\frac{N-p}{2\Lambda}\right) \sum_{a_k\leq i\leq a_{k+1}}X_i X_{i+p}+O(\log n/n^{1/6})\nonumber
\end{eqnarray}
and
\begin{eqnarray}\label{4.10}
&J^3_k=\left(\psi_k+\Lambda\psi^2_k+\frac{N}{4\Lambda}\right)\left([S^n]_{a_{k+1}}- [S^n]_{a_{k}}\right)\\
&+\frac{1}{n}\sum_{p=1}^{N-1}\left(\psi_k+\frac{N-p}{2\Lambda}\right) \sum_{a_k\leq i\leq a_{k+1}}X_i X_{i+p}+O(\log^2 n/n^{1/6}).\nonumber
\end{eqnarray}
The result follows by combining (\ref{4.7+})--(\ref{4.10}).
\end{proof}
\begin{rem}
Lemma \ref{lem4.3} provides some intuition about the term
$\frac{N}{4\Lambda}|S^{\nu}_1-s_0|^2$ which appears in Theorem \ref{thm2.2}.
Indeed, due
to the insider information, the seller can "make money" from the stock price
fluctuation. The corresponding term (in the formulation of Lemma \ref{lem4.3}) is given by
$$\frac{N}{4\lambda}\sum_{k=0}^{m}\left(S^n_{[n\tau^{n,\epsilon}_{k+1}]}-S^n_{[n\tau^{n,\epsilon}_k]}\right)^2$$
and obviously it is closely related to quadratic variation.
As we prove rigorously below, due to
the It\^o Isometry, the asymptotic behavior of these terms is given by $\frac{N}{4\Lambda}|S^{\nu}_1-s_0|^2$.

The main idea in the proof of Lemma 4.4 is to apply trading strategies which
are of the form (\ref{4.7}). We obtained this type of strategies by the "guess and verify" method. Namely, by
applying these strategies, we prove below that the upper bound for the asymptotic behaviour
of the super-replication prices is equal to the lower bound of the super--replication
prices which was established in Section \ref{sec:3}.

A natural question which we leave for future research is the extension of Theorem \ref{thm2.2}
to transaction costs beyond quadratic. Let us notice that different types of
transaction costs require different types of scaling.
\end{rem}

To get a convenient upper bound for $\hat\pi_n(Z^{n,\epsilon,K}/(1-\lambda))$
it will be useful to
consider processes on a slightly expanded time horizon, namely on
$[0,1+\lambda]$ rather than $[0,1]$.
For that purpose
we extend the function $H$ to $H_{1+\lambda}:D[0,1+\lambda] \rightarrow \mathbb{R}_{+}$ simply by letting,
for $p \in D[0,1+\lambda]$,
$$
H_{1+\lambda}(p) := H\left( [0,1] \ni t \mapsto p_{t(1+\lambda)}\right).
$$
Fix $\epsilon,\lambda\in (0,1)$ and let $K=K(\epsilon,\lambda)$ as before. Let
$\mathcal D^{\epsilon,\lambda}$ be the set of stochastic processes $D=\{D_t\}_{t=0}^{1+\lambda}$ defined on some
filtered probability space $(\Omega^D,\mathcal G^D,\{\mathcal G^D_t\}_{t=0}^{1+\lambda},\mathbb P^D)$ (we assume that the filtration is right continuous
and completed),
have the form
\begin{equation}\label{form}
    D_t = \sum_{k=1}^K \Upsilon_{k-1} \mathbb I_{t\in [\theta_{k-1},\theta_k)}+(\Upsilon_{K}+\sigma W_{t-\theta_K})\mathbb I_{t\in [\theta_K,1+\lambda]}
\end{equation}
where $\theta_k$, $k=0,1...,K$ are stopping times (with respect to $\{\mathcal G^D_t\}_{t=0}^{1+\lambda}$)
and
\begin{eqnarray}
&\Upsilon_0=s_0, \ \ |\Upsilon_{k}-\Upsilon_{k-1}| \leq 2\epsilon, \label{4.11}\\
&\theta_0=0, \  \
\frac{\lambda}{K} \leq \theta_k-\theta_{k-1} \leq \frac{\lambda}{K}+\epsilon^2 \label{4.12}\\
&\Upsilon_{k-1}=\mathbb E_{\mathbb P^D}(\Upsilon_k |\mathcal G^D_{\theta_{k-1}}).\label{4.13}
\end{eqnarray}
Moreover, $W$ is a Brownian motion independent of
$\mathcal G^D_{\theta_K}$.
As usual $\{[D]_t\}_{t=0}^{1+\lambda}$ denotes the quadratic variation of $D$ which is given by
\begin{equation*}
[D]_t:=\sum_{k=1}^K
    \sum_{i=1}^{k-1}(\Upsilon_i-\Upsilon_{i-1})^2\mathbb I_{t\in [\theta_{k-1},\theta_k)}
      +\left(\sum_{i=1}^{K}(\Upsilon_i-\Upsilon_{i-1})^2+\overline{\sigma}^2(t-\theta_K)\right)\mathbb I_{t\in [\theta_K,1+\lambda]}.
\end{equation*}
We also define the process $\{\zeta^D_t\}_{t=0}^{1+\lambda}$
\begin{equation}\label{4.zeta}
    \zeta^D_t :=
    \sum_{k=1}^K
     \frac{\mathbb E_{\mathbb P^D}\left(
     \Upsilon^2_k
      -\Upsilon^2_{k-1}
      |\mathcal G^{D}_{\theta_{k-1}}
      \right)}{\mathbb E_{\mathbb P^D}\left(\theta_{k}-\theta_{k-1}|\mathcal G^D_{\theta_{k-1}}\right)}\mathbb I_{t\in [\theta_{k-1},\theta_k)}
      +\overline{\sigma}^2\mathbb I_{t\in [\theta_K,1+\lambda]}.
\end{equation}
We remark that the filtration $\mathcal G^D$  can be larger than the usual filtration generated by $D$.
\begin{prop}\label{prop4.1}
We have the following upper bound
\begin{eqnarray*}
&\lim\sup_{n\rightarrow\infty}\hat\pi_n(Z^{n,\epsilon,K}/(1-\lambda))\leq O(\epsilon+\lambda)\\
&+\sup_{D\in \mathcal D^{\epsilon,\lambda}}\mathbb E_{\mathbb P^D}\left(\frac{H_{1+\lambda}(D)}{1-\lambda}-\frac{N}{4\Lambda}\left(\frac{1}{\lambda^2}\wedge[D]_{1+\lambda}\right)-\frac{1}{16\Lambda}\int_{0}^{1+\lambda}G(\zeta^D_t)dt \right).
\end{eqnarray*}
\end{prop}
\begin{proof}
Choose $n\in\mathbb N$. Recall the set $\Omega_n$ defined after Lemma \ref{lem4.1}, the filtration
$(\mathcal F^{n,\epsilon}_k)_{k\in\mathbb Z_{+}}$ introduced before Lemma \ref{lem4.3} and the set $\mathcal A_m$ introduced
in Lemma \ref{lem4.3}. We take $m=K-1$ and denote $\mathcal A:=\mathcal A_{K-1}$.

Denote by $\hat{\mathcal P}_n$ the set of all
probability measures on $(\Omega_n,\sigma\{X_1,...,X_n\})$.
Introduce the function
$\Gamma:\hat {\mathcal P}_n\times \mathcal A\rightarrow\mathbb R$
by
$$\Gamma(\mathbb P,\psi):=\mathbb E_{\mathbb P}\left(Z^{\psi}-\ln n
\sum_{k=0}^{K-1}\left|\mathbb E_{\mathbb P}\left(S^n_{[n\tau^{n,\epsilon}_{k+1}]}-S^n_{[n\tau^{n,\epsilon}_k]}\big|\mathcal F^{n,\epsilon}_k\right)\right|\right)$$
where
\begin{eqnarray*}
&Z^{\psi}:=\frac{Z^{n,\epsilon,K}}{1-\lambda}-\sum_{k=0}^{K-1}\left(\frac{\psi_k}{2}+\frac{N}{4\Lambda}\right) \left(S^n_{[n\tau^{n,\epsilon}_{k+1}]}-S^n_{[n\tau^{n,\epsilon}_k]}\right)^2\\
&+\sum_{k=0}^{K-1}\left(\frac{\psi_k}{2}+\Lambda\psi^2_k\right)\left([S^n]_{[n\tau^{n,\epsilon}_{k+1}]}-[S^n]_{[n\tau^{n,\epsilon}_{k}]}\right).
\end{eqnarray*}
From Lemma \ref{lem4.3} (for $m=K-1$) we obtain that for sufficiently large $n$
\begin{eqnarray}
&\hat\pi_n(Z^{n,\epsilon,K}/(1-\lambda))\nonumber\\
&<
\epsilon+\inf_{\psi\in\mathcal A}\inf\left\{y:\exists \phi\in\mathcal A: \  y+\sum_{k=0}^{K-1}\phi_k \left(S^n_{[n\tau^{n,\epsilon}_{k+1}]}-S^n_{[n\tau^{n,\epsilon}_k]}\right)\geq Z^{\psi}\right\}\nonumber\\
&=\epsilon+\inf_{\psi\in\mathcal A}\sup_{\mathbb P\in \hat{\mathcal P}_n}\Gamma(\mathbb P,\psi)\label{4.14}
\end{eqnarray}
where the last equality follows from
 classical linear super--replication duality with convexly constrained strategy sets (cf.
\cite{FK}, Theorem 4.1 in connection with Example 2.3). Let us notice that since the
probability space $(\Omega_n,\sigma\{X_1,...,X_n\})$ is finite,
then we can define a probability measure
which gives to every event a positive weight. Hence, pathwise super--replication
is equivalent to almost surely super--replication, and so, we can apply the result from \cite{FK}.

Next, we argue that we can switch the above $\inf$ and $\sup$.
It is readily checked
that for any $\mathbb P\in\hat {\mathcal P}_n$ the map
 $\psi \rightarrow \Gamma(\psi,\mathbb P)$ is convex and that
the map $\mathbb P \rightarrow \Gamma(\psi,\mathbb P)$ is concave for any $\psi$
 fixed. Observing that the sets $\hat {\mathcal P}_n,\mathcal A$ can easily be
 identified with convex and compact subsets in Euclidean space, we can
 thus invoke the Minimax Theorem (e.g. Theorem~45.8 in
 \cite{S}) to obtain
$$\inf_{\psi\in\mathcal A}\sup_{\mathbb P\in \hat{\mathcal P}_n}\Gamma(\mathbb P,\psi)=\sup_{\mathbb P\in \hat{\mathcal P}_n}\inf_{\psi\in\mathcal A}\Gamma(\mathbb P,\psi).$$
From (\ref{4.14}) we conclude that (for sufficiently large $n$)
there exists a probability measure $\mathbb P_n\in \hat{\mathcal P}_n$
for which
\begin{equation}\label{4.15}
\hat\pi_n(Z^{n,\epsilon,K}/(1-\lambda))<\epsilon+\inf_{\psi\in\mathcal A}\Gamma(\mathbb P_n,\psi).
\end{equation}

By combining Corollary \ref{cor3.1}, (\ref{4.new}) and Lemma \ref{lem4.2}
we get
$\lim\inf_{n\rightarrow\infty} \hat\pi_n(Z^{n,\epsilon,K})>-\infty.$
On the other hand, it is straight forward to see that
$\sup_{n\in\mathbb N}||Z^{n,\epsilon,K}||_{\infty}<\infty$.
These observations together with (\ref{4.15}) for $\psi\equiv 0$ give
that for sufficiently large $n$
\begin{equation}\label{4.16}
\mathbb E_{\mathbb P_n}\left(\sum_{k=0}^{K-1}\left|\mathbb E_{\mathbb P_n}\left(S^n_{[n\tau^{n,\epsilon}_{k+1}]}-S^n_{[n\tau^{n,\epsilon}_k]}\big|\mathcal F^{n,\epsilon}_k\right)\right|\right)\leq \frac{1}{\sqrt{\ln n}}.
\end{equation}
Next, for $k=0,1,...,K$ set
\begin{eqnarray*}
&\theta_k:=\tau^{n,\epsilon}_k+\frac{\lambda}{K} k, \\
&\Upsilon_k:=S^n_{[n\tau^{n,\epsilon}_{k}]}-\sum_{i=0}^{k-1}
\mathbb E_{\mathbb P_n}\left(S^n_{[n\tau^{n,\epsilon}_{i+1}]}-S^n_{[n\tau^{n,\epsilon}_i]}\big|\mathcal F^{n,\epsilon}_i\right)
\end{eqnarray*}
and introduce the stochastic processes
\begin{eqnarray*}
&\mathcal N_t=\sum_{k=1}^K\mathbb I_{t\geq \theta_k}, \ \ 0\leq t\leq 1+\lambda\\
&U_t=\sum_{k=1}^K [S^n]_{[n\tau^{n,\epsilon}_{k-1}]}\mathbb I_{t\in [\theta_{k-1},\theta_k)}+[S^n]_{[n\tau^{n,\epsilon}_{K}]}\mathbb{I}_{t\geq\theta_K}, \ \ 0\leq t\leq 1+\lambda.
\end{eqnarray*}
Let $\{D_t\}_{t=0}^{1+\lambda}$ be given by (\ref{form}) and let
$\mathcal G^D_t$ be the filtration generated by
 the processes $\{D_t\}_{t=0}^{1+\lambda},\{\mathcal N_t\}_{t=0}^{1+\lambda},\{U_t\}_{t=0}^{1+\lambda}$ and
 $\{W_{t\vee\theta_K-\theta_K}\}_{t=0}^{1+\lambda}$
 where $W$ is a standard Brownian motion independent of $S^n$.

Let us verify that (\ref{4.11})--(\ref{4.13}) hold true. Indeed,
from the definitions we have $\tau^{n,\epsilon}_{k+1}-\tau^{n,\epsilon}_k\leq \epsilon^2$
and (for sufficiently large $n$) $|S^n_{[n\tau^{n,\epsilon}_{i+1}]}-S^n_{[n\tau^{n,\epsilon}_i]}|\leq 2\epsilon$,
this gives (\ref{4.11}) and (\ref{4.12}) respectively. The equality (\ref{4.13}) follows from the simple
observation that $\mathcal G^D_{\theta_k} =\mathcal F^{n,\epsilon}_k$, $k=0,1...,K$.

We wish to apply (\ref{4.15}) and estimate
$\inf_{\psi\in\mathcal A}\Gamma(\mathbb P_n,\psi)$.
Recall the process $S^{n,\epsilon}$ given by (\ref{4.recall}). It is easy to check that
  \begin{equation}\label{4.16+}
   \max_{k=0,\dots,K}|S^{n,\epsilon}_{\tau^{n,\epsilon}_{k}}-\Upsilon_k|
  \leq
    \sum_{k=0}^{K-1}\left|\mathbb E_{\mathbb P_n}\left(S^n_{[n\tau^{n,\epsilon}_{k+1}]}-S^n_{[n\tau^{n,\epsilon}_k]}\big|\mathcal F^{n,\epsilon}_k\right)\right|
\end{equation}
  and
$\max_{k=0,\dots,K}\left|\theta_k-\tau^{n,\epsilon}_k\right| \leq \lambda.$
  Therefore, on the event $\{\tau^{n,\epsilon}_K=1\}$ we can estimate the Skorohod distance (on the space $D[0,1]$)
  $$
   d(S^{n,\epsilon},D_{(1+\lambda)\cdot .}) \leq \lambda + \sum_{k=0}^{K-1}\left|\mathbb E_{\mathbb P_n}\left(S^n_{[n\tau^{n,\epsilon}_{k+1}]}-S^n_{[n\tau^{n,\epsilon}_k]}\big|\mathcal F^{n,\epsilon}_k\right)\right|.
$$
 This together with (\ref{4.16}) and the fact that $H$ is nonnegative and Lipschitz continuous in the Skorohod metric yields
 that for sufficiently large $n$
 \begin{equation}\label{4.17}
 \mathbb E_{\mathbb P_n}\left(\frac{Z^{n,\epsilon,K}}{1-\lambda}\right) <
   \mathbb E_{\mathbb P_n}\left(\frac{H_{1+\lambda}(D)}{1-\lambda}\right)+O(\lambda).
    \end{equation}
    Next, we treat the term
    $\sum_{k=0}^{K-1}\left(S^n_{[n\tau^{n,\epsilon}_{k+1}]}-S^n_{[n\tau^{n,\epsilon}_k]}\right)^2$.
    From (\ref{estimate}), on the event $\{\tau^{n,\epsilon}_K<1\}$ we have (recall that $c(\lambda)>2$)
    $$\sum_{k=0}^{K-1}\left(S^n_{[n\tau^{n,\epsilon}_{k+1}]}-S^n_{[n\tau^{n,\epsilon}_k]}\right)^2\geq K\epsilon^2-1\geq \frac{1}{\lambda^2}.$$
    On the other hand, on the event $\{\tau^{n,\epsilon}_K=1\}$ we have $\theta_K=1+\lambda$ and so,
    on this event
    $[D]_{1+\lambda}=\sum_{i=1}^{K}(\Upsilon_i-\Upsilon_{i-1})^2.$
    From (\ref{4.16})--(\ref{4.16+}) we conclude that for sufficiently large $n$
    \begin{equation}\label{4.177}
    \mathbb E\left(\sum_{k=0}^{K-1}\left(S^n_{[n\tau^{n,\epsilon}_{k+1}]}-S^n_{[n\tau^{n,\epsilon}_k]}\right)^2\right)>
    \mathbb E_{\mathbb P^D}\left(\frac{1}{\lambda^2}\wedge [D]_{1+\lambda}\right)-\epsilon.
   \end{equation}
    We arrive to the final step in the estimation
$\inf_{\psi\in\mathcal A}\Gamma(\mathbb P_n,\psi)$.
Fix $k=0,1,...,K-1$ and set
$$\Phi_k:=\max\left(\underline{\sigma}^2,\frac{\mathbb E_{\mathbb P_n}\left([S^n]_{[n\tau^{n,\epsilon}_{k+1}]}-[S^n]_{[n\tau^{n,\epsilon}_{k}]}\big|\mathcal F^{n,\epsilon}_k\right)}{\mathbb E_{\mathbb P_n}\left(\frac{\epsilon}{K}+
\tau^{n,\epsilon}_{k+1}-\tau^{n,\epsilon}_{k}\big|\mathcal F^{n,\epsilon}_k\right)}\right).$$
 From the inequality $|X_i|\geq\underline{\sigma}$, $i\in\mathbb N$, it follows that for sufficiently large $n$ we have
$$0\leq
\Phi_k\mathbb E_{\mathbb P_n}\left(\frac{\epsilon}{K}+
\tau^{n,\epsilon}_{k+1}-\tau^{n,\epsilon}_{k}\big|\mathcal F^{n,\epsilon}_k\right)-
\mathbb E_{\mathbb P_n}\left([S^n]_{[n\tau^{n,\epsilon}_{k+1}]}-[S^n]_{[n\tau^{n,\epsilon}_{k}]}\big|\mathcal F^{n,\epsilon}_k\right)\leq
O\left(\frac{\epsilon}{K}\right).$$
Hence,
\begin{eqnarray*}
&\max_{\psi\in\mathcal A}\mathbb E_{\mathbb P_n}
\left(\frac{\psi_k}{2}\left(S^n_{[n\tau^{n,\epsilon}_{k+1}]}-S^n_{[n\tau^{n,\epsilon}_k]}\right)^2
-\left(\frac{\psi_k}{2}+\Lambda\psi^2_k\right)\left([S^n]_{[n\tau^{n,\epsilon}_{k+1}]}-[S^n]_{[n\tau^{n,\epsilon}_{k}]}\right)\right)\\
&\geq O\left(\frac{\epsilon}{K}\right)\min_{\beta\in\mathbb R}\left(\frac{\beta}{2}+\Lambda\beta^2\right)\\
&+\max_{\psi\in\mathcal A}\mathbb E_{\mathbb P_n}\Bigg(\frac{\psi_k}{2}\mathbb E_{\mathbb P_n}\left(\left(S^n_{[n\tau^{n,\epsilon}_{k+1}]}-S^n_{[n\tau^{n,\epsilon}_k]}\right)^2\big|\mathcal F^{n,\epsilon}_k\right)\\
&-\Phi_k \mathbb E_{\mathbb P_n}\left(\frac{\epsilon}{K}+
\tau^{n,\epsilon}_{k+1}-\tau^{n,\epsilon}_{k}\big|\mathcal F^{n,\epsilon}_k\right)\left(\frac{\psi_k}{2}+\Lambda\psi^2_k\right)\Bigg).
\end{eqnarray*}
The above expression is a quadratic pattern in $\psi_k$ which take maximum for
$$\psi^{*}_k=
\frac{\mathbb E_{\mathbb P_n}\left(\left(S^n_{[n\tau^{n,\epsilon}_{k+1}]}-S^n_{[n\tau^{n,\epsilon}_k]}\right)^2\big|\mathcal F^{n,\epsilon}_k\right)}{4\Lambda\Phi_k\mathbb E_{\mathbb P_n}\left(\frac{\epsilon}{K}+
\tau^{n,\epsilon}_{k+1}-\tau^{n,\epsilon}_{k}\big|\mathcal F^{n,\epsilon}_k\right)}-\frac{1}{4\Lambda}.$$
Clearly, for sufficiently large $n$, $|\psi^{*}_k|\leq\ln n$.

Recall the definition of $\zeta^D$ given in (\ref{4.zeta}). We obtain,
\begin{eqnarray}
&\max_{\psi\in\mathcal A}\sum_{k=0}^{K-1}\mathbb E_{\mathbb P_n}
\bigg(\frac{\psi_k}{2}\left(S^n_{[n\tau^{n,\epsilon}_{k+1}]}-S^n_{[n\tau^{n,\epsilon}_k]}\right)^2\nonumber\\
&-\left(\frac{\psi_k}{2}+\Lambda\psi^2_k\right)\left([S^n]_{[n\tau^{n,\epsilon}_{k+1}]}-[S^n]_{[n\tau^{n,\epsilon}_{k}]}\right)\bigg)\nonumber\\
&\geq\sum_{k=0}^{K-1}\mathbb E_{\mathbb P_n}\left(\frac{\left(
\mathbb E_{\mathbb P_n}\left(\left(S^n_{[n\tau^{n,\epsilon}_{k+1}]}-S^n_{[n\tau^{n,\epsilon}_k]}\right)^2\big|\mathcal F^{n,\epsilon}_k\right)
-\Phi_k \mathbb E_{\mathbb P_n}\left(\frac{\epsilon}{K}+
\tau^{n,\epsilon}_{k+1}-\tau^{n,\epsilon}_{k}\big|\mathcal F^{n,\epsilon}_k\right)\right)^2}{16\Lambda\Phi_k \mathbb E_{\mathbb P_n}\left(\frac{\epsilon}{K}+
\tau^{n,\epsilon}_{k+1}-\tau^{n,\epsilon}_{k}\big|\mathcal F^{n,\epsilon}_k\right)}\right)\nonumber\\
&-O(\epsilon)\nonumber\\
&\geq \frac{1}{16\Lambda}\sum_{k=0}^{K-1}\mathbb E_{\mathbb P_n}\left(G\left(\frac{\mathbb E_{\mathbb P_n}\left(\left(S^n_{[n\tau^{n,\epsilon}_{k+1}]}-S^n_{[n\tau^{n,\epsilon}_k]}\right)^2\big|\mathcal F^{n,\epsilon}_k\right)}{\mathbb E_{\mathbb P_n}\left(\frac{\epsilon}{K}+
\tau^{n,\epsilon}_{k+1}-\tau^{n,\epsilon}_{k}\big|\mathcal F^{n,\epsilon}_k\right)}\right)\left(\frac{\epsilon}{K}+
\tau^{n,\epsilon}_{k+1}-\tau^{n,\epsilon}_{k}\right)\right)\nonumber\\
&-O(\epsilon)\nonumber\\
&\geq \frac{1}{16\Lambda}\sum_{k=0}^{K-1}\mathbb E_{\mathbb P_n}\left(G\left(\zeta^D_{\theta_k}\right)\left(\theta_{k+1}-\theta_k\right)\right)-O(\epsilon)\nonumber\\
&=
\frac{1}{16\Lambda}\mathbb E_{\mathbb P_n}\left(\int_{0}^{1+\lambda}G(\zeta^D_t) dt\right)-O(\epsilon).\label{4.18}
\end{eqnarray}
The first inequality follows from the above analysis. The second inequality follows from
the fact that for sufficiently large $n$,
$\Phi_k\in [\underline{\sigma}^2,\overline{\sigma}^2]$, $k=0,1,...,K-1$ and
the relation $G(z)=\min_{y\in [\underline{\sigma}^2,\overline{\sigma}^2]}\frac{(z-y)^2}{y}$, $z\in\mathbb R$. We also used the law of total expectation.
The third inequality is due to (\ref{4.16})--(\ref{4.16+})
and the fact that $\Upsilon_k, S^n_{[n\tau^{n,\epsilon}_k]}$, $k=0,1...,K$
are uniformly bounded (in $n$). The last equality is trivial.

The proof follows from combining (\ref{4.15}), and (\ref{4.17})--(\ref{4.18}).
\end{proof}
In view of Lemma \ref{lem4.2} and Proposition \ref{prop4.1}, in order to prove (\ref{4.3}) (and complete the proof of Proposition \ref{prop.campi2})
it remains to establish the following limit theorem.
\begin{lem}
\begin{eqnarray*}
&{\underline{\lim}}_{\lambda\rightarrow 0}{\underline{\lim}}_{\epsilon\rightarrow 0}\sup_{D\in \mathcal D^{\epsilon,\lambda}}\mathbb E_{\mathbb P^D}\left(\frac{H_{1+\lambda}(D)}{1-\lambda}-\frac{N}{4\Lambda}\left(\frac{1}{\lambda^2}\wedge[D]_{1+\lambda}\right)-\frac{1}{16\Lambda}\int_{0}^{1+\lambda}G(\zeta^D_t)dt \right)\\
&\leq\sup_{\nu\in\Gamma}\mathbb E_{\mathbb P^W}\left(H(S^{\nu})-\frac{N}{4\Lambda}|S^{\nu}_1-s|^2-\frac{1}{16\Lambda}\int_{0}^1 G(\nu^2_t) dt\right).
\end{eqnarray*}
\end{lem}
\begin{proof}
Fix $\lambda>0$. For any $m\in\mathbb N$ let $\epsilon_m:=\frac{1}{m}$ and choose
$D^m\in\mathcal D^{\epsilon_m,\lambda}$
such that
\begin{eqnarray*}
&\frac{1}{m}+\mathbb E_{\mathbb P^{D^m}}\left(\frac{H_{1+\lambda}(D^m)}{1-\lambda}-\frac{N}{4\Lambda}\left(\frac{1}{\lambda^2}\wedge[D]_{1+\lambda}\right)-\frac{1}{16\Lambda}\int_{0}^{1+\lambda}G(\zeta^{D^m}_t)dt \right)\\
&\geq\sup_{D\in \mathcal D^{\epsilon_m,\lambda}}\mathbb E_{\mathbb P^{D}}\left(\frac{H_{1+\lambda}(D)}{1-\lambda}-\frac{N}{4\Lambda}\left(\frac{1}{\lambda^2}\wedge[D]_{1+\lambda}\right)-\frac{1}{16\Lambda}\int_{0}^{1+\lambda}G(\zeta^D_t)dt \right).
\end{eqnarray*}
Following the same arguments as in Lemma 3.9 in \cite{BD} gives that there exists a subsequence (which still denoted by $m$) and a continuous martingale
$M^{\lambda}=\{M^{\lambda}_t\}_{t=0}^{1+\lambda}$ such that we have the weak convergence
$$ \left(D^{m},[D^m],\int_{0}^{\cdot}\zeta^{D^m}_t dt\right)\Rightarrow \left(M^{\lambda},\langle  M^{\lambda}\rangle,\langle  M^{\lambda}\rangle\right) \ \
\mbox{on} \ \  D^3[0,1+\lambda].$$
As usual $\langle M^{\lambda}\rangle $ denotes the quadratic variation of the continuous martingale $M^{\lambda}$.
From the Skorohod representation theorem (see \cite{D}) it follows that we can construct a probability space
(the corresponding probability measure will be denoted by $\mathbb P$)
 such that
\begin{equation}\label{4.201}
 \left(D^{m},[D^m],\int_{0}^{\cdot}\zeta^{D^m}_t dt\right)\rightarrow \left(M^{\lambda},\langle  M^{\lambda}\rangle,\langle  M^{\lambda}\rangle\right) \ \ \mathbb{P} \ \mbox{a.s.}
\end{equation}
From  (\ref{4.11})--(\ref{4.zeta}) it follows that
\begin{equation}\label{4.44}
\sup_{m\in\mathbb N}\sup_{0\leq t\leq {1+\lambda}}\zeta^{D^m}_t<\overline{\sigma}^2+\sup_{m\in\mathbb N}\frac{4 K(1/m,\lambda)}{\lambda m^2}<\infty.
\end{equation}
Thus,
by Lemma A1.1 in \cite{DS1}
 we construct a sequence $\alpha^{m}\in conv(\zeta^{D^m},\zeta^{D^{m+1}},...)$, $m\in\mathbb N$ such that
$\alpha^{m}$ converge $\mathbb P\otimes dt$ almost surely to a stochastic process $\alpha$.
From (\ref{4.201})--(\ref{4.44}) and the bounded convergence theorem it follows that
$$\int_{0}^t \alpha_u du=\lim_{m\rightarrow\infty}\int_{0}^t\alpha^{m}_u du=\lim_{m\rightarrow\infty}\int_{0}^t\zeta^{D^m}_u du=\langle M^{\lambda}\rangle_t, \ \
\mathbb P\otimes dt \ \mbox{a.s.}$$
We conclude
that $\alpha=\frac{d\langle M^{\lambda}\rangle}{dt}$ $\mathbb P\otimes dt$ a.s.

From the Fatou lemma and the convexity of $G$ we obtain
\begin{eqnarray}
&\mathbb E_{\mathbb P}\left(\int_{0}^{1+\lambda} G\left(\frac{d\langle M^{\lambda}\rangle}{dt}\right) dt\right)\nonumber \\
&\leq\lim\inf_{m\rightarrow\infty}\mathbb E_{\mathbb P}\left(\int_{0}^{1+\lambda} G(\alpha^m_t) dt\right)\nonumber \\
&\leq \lim\inf_{m\rightarrow\infty}\mathbb E_{\mathbb P^{D^m}}\left(\int_{0}^{1+\lambda} G(\zeta^{D^m}_t) dt\right).\label{4.202}
\end{eqnarray}
From the dominated convergence theorem and
(\ref{4.201})
\begin{equation}\label{4.203}
\lim_{m\rightarrow\infty}\mathbb E_{\mathbb P^{D^m}}\left(\frac{1}{\lambda^2}\wedge[D^m]_{1+\lambda}\right)=\mathbb E_{\mathbb P}\left(\frac{1}{\lambda^2}\wedge\langle M\rangle_{1+\lambda}\right).
\end{equation}
From
the Doob--Kolmogorov inequality, the
linear growth of $H$ and (\ref{4.44}) we conclude that
the random variables $\{H_{1+\lambda}(D^m)\}_{m\in\mathbb N}$ are uniformly integrable. This together with
(\ref{4.201}) gives
\begin{equation}\label{4.204}
\lim_{m\rightarrow\infty}\mathbb E_{\mathbb P^{D^m}}\left(H_{1+\lambda}(D^{m})\right)=\mathbb E_{\mathbb P}\left(H_{1+\lambda}(M^{\lambda})\right).
\end{equation}
By combining (\ref{4.202})--(\ref{4.204}) we arrive to
\begin{eqnarray*}
&{\underline{\lim}}_{\lambda\rightarrow 0}{\underline{\lim}}_{\epsilon\rightarrow 0}\sup_{D\in \mathcal D^{\epsilon,\lambda}}\mathbb E_{\mathbb P^D}\left(\frac{H_{1+\lambda}(D)}{1-\lambda}-\frac{N}{4\Lambda}\left(\frac{1}{\lambda^2}\wedge[D^m]_{1+\lambda}\right)-\frac{1}{16\Lambda}\int_{0}^{1+\lambda}G(\zeta^D_t)dt \right)\nonumber\\
&\leq {\underline{\lim}}_{\lambda\rightarrow 0}\mathbb E_{\mathbb P^D}\left(\frac{H_{1+\lambda}(M^{\lambda})}{1-\lambda}-\frac{N}{4\Lambda}\left(\frac{1}{\lambda^2}\wedge\langle M^{\lambda}\rangle_{1+\lambda}\right)-\frac{1}{16\Lambda}\int_{0}^{1+\lambda} G\left(\frac{d\langle M^{\lambda}\rangle}{dt}\right)dt \right).
\end{eqnarray*}
Finally, it remains to establish the inequality
\begin{eqnarray}\label{4.206}
&{\underline{\lim}}_{\lambda\rightarrow 0}\mathbb E_{\mathbb P^D}\left(\frac{H_{1+\lambda}(M^{\lambda})}{1-\lambda}-\frac{N}{4\Lambda}\left(\frac{1}{\lambda^2}\wedge\langle M^{\lambda}\rangle_{1+\lambda}\right)-\frac{1}{16\Lambda}\int_{0}^{1+\lambda} G\left(\frac{d\langle M^{\lambda}\rangle}{dt}\right)dt \right)\nonumber\\
&\leq\sup_{\nu\in\Gamma}\mathbb E_{\mathbb P^W}\left(H(S^{\nu})-\frac{N}{4\Lambda}|S^{\nu}_1-s|^2-\frac{1}{16\Lambda}\int_{0}^1 G(\nu^2_t) dt\right).
\end{eqnarray}
Clearly, the right hand side of (\ref{4.206}) is bigger than $-\frac{\underline{\sigma}^2}{16\Lambda}$
(just take $\nu\equiv 0$).
 Thus, without loss of generality we can assume that for sufficiently small $\lambda>0$ we have
\begin{equation}\label{final}
\mathbb E_{\mathbb P}\left(\frac{H_{1+\lambda}(M^{\lambda})}{1-\lambda}-\frac{N}{4\Lambda}\left(\frac{1}{\lambda^2}\wedge\langle M^{\lambda}\rangle_{1+\lambda}\right)-\frac{1}{16\Lambda}\int_{0}^{1+\lambda} G\left(\frac{d\langle M^{\lambda}\rangle}{dt}\right)dt \right)\geq -\frac{\underline{\sigma}^2}{8\Lambda}.
\end{equation}
Otherwise, (\ref{4.206}) is trivial.
From the Burkholder--David--Gundy inequality and the linear Growth of $H$ it follows that there exists a constant $C_1>0$ such that
$$\mathbb E_{\mathbb P}\left(\frac{H_{1+\lambda}(M^{\lambda})}{1-\lambda}\right)\leq C_1\left(1+\langle M^{\lambda}\rangle_{1+\lambda} \right).
$$
On the other hand it is straight forward to see that there exists a constant $C_2>0$ such that
$G(z)\geq \frac{z^2}{C_2}-1$, $\forall z\geq 0$. From (\ref{final}) we conclude that
\begin{equation}\label{volbound}
\lim\sup_{\lambda\rightarrow 0} \mathbb{E}_{\mathbb P}\left(\int_{0}^{1+\lambda}\left(\frac{d\langle  M^\lambda\rangle _t}{dt}\right)^2 dt\right)<\infty.
\end{equation}
Finally, for any $\lambda\in (0,1)$ define the continuous martingale $\bar M^{\lambda}$ on the time interval $[0,1]$ by
$\bar M^{\lambda}_t:=M^{\lambda}_{(1+\lambda)t}$, $t\in [0,1]$.
Observe that
$\bar M^{\lambda}_1= M^{\lambda}_{1+\lambda}$.
From the Cauchy--Schwarz inequality, the It\^o Isometry and (\ref{volbound}) we obtain
\begin{eqnarray*}
&\lim\inf_{\lambda\rightarrow 0}
\mathbb E_{\mathbb P}\left(\frac{H_{1+\lambda}(M^{\lambda})}{1-\lambda}-\frac{N}{4\Lambda}\left(\frac{1}{\lambda^2}\wedge\langle M^{\lambda}\rangle_{1+\lambda}\right)-\frac{1}{16\Lambda}\int_{0}^{1+\lambda} G\left(\frac{d\langle M^{\lambda}\rangle}{dt}\right)dt \right)\\
&=
\lim\inf_{\lambda\rightarrow 0}
\mathbb E_{\mathbb P}\left(H_{1+\lambda}(\bar M^{\lambda})-\frac{N}{4\Lambda}|\bar M^{\lambda}_{1}-s|^2-\frac{1}{16\Lambda}\int_{0}^{1} G\left(\frac{d\langle \bar M^{\lambda}\rangle}{dt}\right)dt \right)\\
&\leq\sup_{\nu\in\Gamma}\mathbb E_{\mathbb P^W}\left(H(S^{\nu})-\frac{N}{4\Lambda}|S^{\nu}_1-s|^2-\frac{1}{16\Lambda}\int_{0}^1 G(\nu^2_t) dt\right)
\end{eqnarray*}
where the last inequality follows by using the same arguments (based on the randomization technique) as in Lemma 7.2 in \cite{DS}.
This completes the proof.
\end{proof}
\begin{rem}\label{rem3}
We finish with claiming that the upper bound holds for payoffs
of the from given by (\ref{formm}).
Clearly, Theorem \ref{thm2.2} remains true for claims  which are bounded from below and Lipschitz
  continuous with respect to the Skorohod metric. Thus, for any $m\in\mathbb N$ we have
\begin{eqnarray*}
&\lim_{n\rightarrow\infty} \pi_n\left(\hat H_m\left(\{S^{n}_{[nt]}\}_{t=0}^1\right)\right)\\
&=\sup_{\nu\in\Gamma}\mathbb E_{\mathbb P^W}\left(\hat H_m(S^{\nu})-\frac{N}{4\Lambda}|S^{\nu}_1-s|^2-\frac{1}{16\Lambda}\int_{0}^1 G(\nu^2_t) dt\right)
\end{eqnarray*}
where $\hat H_m$ is given by
$$\hat H_m(p):=H(p)-\alpha (m\wedge |p_1-p_0|^2), \ \ \forall p\in D[0,1].$$
Since $\hat H_m\geq \hat H$ then
$\pi_n\left(\hat H_m\left(\{S^{n}_{[nt]}\}_{t=0}^1\right)\right)\geq \pi_n\left(\hat H\left(\{S^{n}_{[nt]}\}_{t=0}^1\right)\right)$ for all $n,m\in\mathbb N$.
Thus, by using same arguments as in the proof of the inequality (\ref{4.206}) above,
we obtain
\begin{eqnarray*}
&\lim_{n\rightarrow\infty} \pi_n\left(\hat H\left(\{S^{n}_{[nt]}\}_{t=0}^1\right)\right)\\
&\leq\lim_{m\rightarrow\infty} \sup_{\nu\in\Gamma}\mathbb E_{\mathbb P^W}\left(\hat H_m(S^{\nu})-\frac{N}{4\Lambda}|S^{\nu}_1-s|^2-\frac{1}{16\Lambda}\int_{0}^1 G(\nu^2_t) dt\right)\\
&=\sup_{\nu\in\Gamma}\mathbb E_{\mathbb P^W}\left(\hat H(S^{\nu})-\frac{N}{4\Lambda}|S^{\nu}_1-s|^2-\frac{1}{16\Lambda}\int_{0}^1 G(\nu^2_t) dt\right)
\end{eqnarray*}
as required.
\end{rem}

\section*{Acknowledgments}
This research was supported by the
GIF Grant 1489-304.6/2019 and the ISF grant 160/17.

\bibliographystyle{spbasic}

\end{document}